\documentclass[a4paper]{amsproc}
\usepackage{amssymb}
\usepackage{amscd}
\usepackage{pstricks}
\usepackage[dvips]{graphicx}
\usepackage{pst-3dplot}
\usepackage[all,cmtip]{xy}

\theoremstyle{plain}
 \newtheorem{thm}{Theorem}[section]
 \newtheorem{prop}{Proposition}[section]
 \newtheorem{lem}{Lemma}[section]
 
\theoremstyle{definition}
 
 \newtheorem{remark}{Remark}[section]

\numberwithin{equation}{section}

\newcommand{\R}{\mathbb{ R}}

\DeclareMathOperator{\diag}{diag}

\DeclareMathOperator{\tr}{tr}
\DeclareMathOperator{\grad}{grad}

\title[Nonholonomic connections and integrability of the rolling ball]{NONHOLONOMIC CONNECTIONS, TIME REPARAMETRIZATIONS,  AND INTEGRABILITY OF THE ROLLING BALL OVER A SPHERE}

\author{Borislav Gaji\' c and Bo\v zidar Jovanovi\' c}


\address{\sc Mathematical Institute SANU, Kneza
Mihaila 36, 11000, Belgrade, Serbia}
\email{gajab@mi.sanu.ac.rs, bozaj@mi.sanu.ac.rs}

\begin{document}

\begin{abstract}
We study a time reparametrisation of the Newton type equations on Riemannian manifolds
slightly modifying the Chaplygin multiplier method,
allowing us to consider the Chaplygin method and the Maupertuis principle within a unified framework.
As an example, the reduced nonholonomic problem of rolling without
slipping and twisting of an $n$-dimensional balanced ball over a fixed sphere is considered. For a special inertia operator (depending on $n$ parameters) we prove complete integrability when the radius of the ball is twice the radius of the sphere. In the case of $SO(l)\times SO(n-l)$ symmetry, noncommutative integrability for any ratio of the radii is established.
\end{abstract}

\maketitle

\section{Introduction}

Let $(M,L,\mathcal D)$ be a nonholonomic Lagrangian system, where
$M$ is a $n$-dimensional manifold, $L: TM\to \R$ Lagrangian, and $\mathcal D$ nonintegrable
$(n-k)$-dimensional distribution of constraints, locally defined by 1-forms $\alpha^a$, $a=1,\dots, k$.  Let
$q=(q^1,\dots,q^n)$ be local coordinates on $M$ in which the
constraints are written in the form
\begin{equation}
(\alpha^a,\dot q)=\sum_{i=1}^n \alpha_i^a(q)\dot q^i=0,\qquad
a=1,\dots,k. \label{constraints}\end{equation}

The curves $q(t)$ such that the velocity $\dot{q}$ belongs to the distribution $\mathcal D$ are called \emph{admissible}.
The motion of  the system is described by the Lagrange-d'Alembert principle
\begin{equation}
\Big(\frac{d}{dt}\frac{\partial L}{\partial \dot q}-\frac{\partial
L}{\partial q},\xi\Big)=\sum_i\Big(\frac{d}{dt}\frac{\partial L}{\partial \dot q^i}-\frac{\partial
L}{\partial q^i}\Big)\xi^i=0 \quad \text{for all} \quad\xi\in \mathcal D_q,
\label{L-dAlembert}
\end{equation}
i.e.,
\begin{equation}
\frac{d}{dt}\frac{\partial L}{\partial \dot q^i}=\frac{\partial
L}{\partial q^i}+ \sum_{a=1}^{k}\lambda_a \alpha^a_i,  \qquad
i=1,\dots,n,\label{L-dAlembert2}
\end{equation}
where the Lagrange multipliers $\lambda_a$ are chosen such that the
solutions $q(t)$ are admissible curves.
The sum $\sum_{a=1}^{k}\lambda_a \alpha^a_i$
represents the reaction force of the constraints.

For natural mechanical systems the Lagrangian is
the difference between the kinetic $T=\frac12\langle \dot q,\dot q\rangle=\frac12\sum_{ij} g_{ij} \dot q^i\dot q^j$ and the potential energy $V(q)$,
and we have a well known relation between the variational derivative of the Lagrangian and
the Levi-Civita connection $\nabla$ of the metric $g$:
\begin{equation}\label{veza}
g^{-1}\Big(\frac{d}{dt}\frac{\partial L}{\partial \dot q}-\frac{\partial
L}{\partial q}\Big)=\nabla_{\dot q} \, \dot q+\mathrm{grad}\, V(q),
\end{equation}
where the metric tensor $g$ is considered as a mapping $g\colon T M\to T^* M$.\footnote{Through the paper we use  $\langle \cdot,\cdot\rangle$, $\langle \cdot,\cdot\rangle_0$, etc. in calculating the scalar products defined by the metrics $g$, $g_0$, etc., while $(\cdot,\cdot)$ denotes the Euclidean scalar product in $\R^n$.}
Thus, \eqref{L-dAlembert} can be written in the form
\begin{equation}\label{NJ}
\langle \nabla_{\dot q} \,\dot q+\mathrm{grad}\, V(q),\xi\rangle =0, \qquad \dot q, \xi\in\mathcal D_q.
\end{equation}

Considering \eqref{NJ}, it is natural to define \emph{connection of the vector bundle} $\mathcal D \rightarrow M$:
\begin{equation}\label{nablap}
\nabla^P\colon \Gamma(TM)\times \Gamma(\mathcal D)\longrightarrow \Gamma(\mathcal D), \qquad  \nabla^P_X Y:=P(\nabla_X Y),
\end{equation}
$X\in\Gamma(TM)$, $Y\in\Gamma(\mathcal D)$, where $P$ is the orthogonal projection to $\mathcal D$.
It is a metric connection, in a sense that
\[
Z\langle X,Y\rangle=\langle \nabla^P_Z X,Y\rangle+\langle X,\nabla^P_Z Y\rangle, \qquad Z\in\Gamma(TM), \quad X,Y\in\Gamma(\mathcal D).
\]
Thus, the parallel transport is an isometry along the distribution $\mathcal D$.
Then \eqref{NJ} is equivalent to
\begin{equation}\label{NHJ}
\nabla^P_{\dot q} \,\dot q={\grad}_\mathcal D V(q),\qquad \dot q\in\mathcal D_q,
\end{equation}
where $\grad_\mathcal D V=P(\grad V)$.
If $V\equiv 0$, \eqref{NHJ} takes the form of
the equation of geodesic lines (so called \emph{nonholonomic geodesics})
\begin{equation}\label{GL}
\nabla^P _{\dot q} \,\dot q=0, \qquad \dot q\in\mathcal D_q.
\end{equation}

Hertz was first who observed that the solutions of nonholonomic problems are not solutions of the variational problem: they are not shortest, but straitness lines.
In classical works of Synge \cite{Sy}, Vranceanu \cite{Vr}, Shouten \cite{Sh}, Wagner \cite{Wa1, Wa2} the problem of motion of nonholonomic systems from the geometric point of view is considered.
The analogous variational problem is studied within a framework of sub-Riemannian  geometry or vaconomic mechanics, which will not be considered here.

The paper is organize as follows.  In Section 2 we recall on
the extensions of the vector-bundle connection \eqref{nablap}
to the linear connection on $TM$ considered in \cite{VeFa, BC} and \cite{Lew}, as well as on so called \emph{partial connection} (e.g., see \cite{DrGa}).
Although mentioned objects are very well studied, some natural relationships between them are pointed out.
In the extension given in \cite{VeFa, BC}, the constraints \eqref{constraints} are integrals of the associated geodesic flows, while in
the extension given in \cite{Lew} they are invariant relations (Lemma \ref{lema}).
 In Proposition \ref{pomocna} vector bundle connections of natural mechanical systems having the same dynamics are compared,
 motivating the study of the partial connection and its geometry developed by Shouten and Wagner.
For the completeness of the exposition and in order to fix the notation for the main results, in Section 3 we further
recall basic facts about $G$--Chaplygin systems and the reduction of the connection following Koiller \cite{Koi} and Bak\v sa \cite{Ba} (for Abelian systems).

In Section 4 we consider the Newton type equations on a Riemannian manifold $(M,g)$ and look for a conformal metric $g_*=f^2g$ such that
solutions of the Newton equations, after a time reparametrization, become the geodesic lines of $g_*$ (Proposition \ref{stav}).
The idea is quite simple, but its slightly generalise the method Chaplygin multiplier for Hamiltonization of $G$-Chaplygin systems
\cite{BBM, Co, CCLM, EKR, FeJo, St}. Also, we obtain variants of the Maupertuis principle in nonholonomic mechanics as they are given in \cite{Koi, Ba}.

As an example,  the nonholonomic problem of rolling without
slipping and twisting of an $n$-dimensional balanced ball of radius $\rho$ over a fixed sphere of radius $\sigma$ is considered (Section 5).
It is a $SO(n)$--Chaplygin system that reduces to the tangent bundle of a $(n-1)$-dimensional sphere  (see \cite{Jo6}).
Remarkably, for a special choice of inertia operator,
the reduced system is suitable for both time reparametrisations considered in Section 4. Combining the Chaplygin reducing multiplier (see \cite{Jo6}) and the Maupertuis principle, we transform the reduced system to the zero-energy level set of a natural mechanical system on a sphere endowed with the standard metric and the potential $V_{\epsilon}(x)=-(A^{-1}x,x)^{-\frac{1}{\epsilon}}$, where  $\epsilon={\sigma}/({\sigma\pm\rho})$ and $A=\diag(a_1,\dots,a_n)$
(see Proposition \ref{jakobi}).
In particular, for $\epsilon=+1$, we have Braden's potential \cite{Br}, and for
$\epsilon=-1$ a celebrated Neumann's potential \cite{Moser}.

Thus,
we get the complete integrability when the radius of the ball is twice the radius of the sphere and the ball is a spherical
shall that rolls over a sphere placed inside ($\epsilon=-1$, Theorem \ref{glavna}).
In the three-dimensional case Borisov and Mamaev \cite{BM2} proved integrability of the problem and constructed separating variables using a variant of sphero-conical coordinates (also called  a \emph{nonholonomic deformation of sphero-conical coordinates} or \emph{quasi-sphero-conical coordinates}, see  \cite{BMF, Ts, Ts2}). We obtain, in some sense,
 an explanation for unusual choice of variables by Borisov and Mamaev (see Subsection \ref{razdvajanje}).
Further, the case $\epsilon=+1$ is the limit case, when the radius of the fixed sphere tends to infinity, and we get an alternative proof (Theorem \ref{veselovaRevised}) of the integrability of the Veselova problem studied in \cite{FeJo}.

Finally, apart of the case when $A$ is proportional to the identity matrix and the trajectories of the reduced system are great circles for all $\epsilon$, we also prove integrability in the case where the matrix $A$ has only two distinct parameters
(Theorem \ref{simetricna}).

\section{Connections in noholonomic mechanics}

One of natural questions is a construction of the extension of $\nabla^P$ to an affine connection $\nabla$ on $TM$ such that $\nabla^P=\nabla\vert_{\Gamma(TM)\times\Gamma(\mathcal D)}$.
 Since the distribution $\mathcal{D}$ is nonintegrable, any such extension has a
non-vanishing torsion, for $X,Y\in\Gamma(\mathcal D)$ given by
\begin{equation}\label{torsion}
T(X,Y)=\nabla^P_X\,Y - \nabla^P_Y\, X -[X,Y]=-Q[X,Y].
\end{equation}

Obviously, an extension of the vector bundle
connection $\nabla^P$ is not unique. We recall on extensions that are given in \cite{VeFa, Lew}, as well as on a restriction of $\nabla^P$ to the partial connection $\nabla^\mathcal D$ on $\mathcal D$ (e.g, see \cite{DrGa}).

\subsection{Extensions of the vector bundle connections}

Vershik and Fadeev \cite{VeFa} defined an affine connection on $TM$ as follows.
Let $A^a=g^{-1}(\alpha^a)$, i.e.,
\[
A^{ai}=\sum g^{ij}\alpha^a_j, \qquad a=1,\dots,k,
\]
 and let $Q$ be the orthogonal projection to $\mathcal D^\perp$.
Then the constraints take the form $\alpha^a(\dot q)=\langle A^a,\dot q\rangle=0$, $a=1,\dots, k$,
and the orthogonal projections $Q$ and $P$ read
\begin{align*}
&Q(X)=\sum_{ab} a_{ab}^{-1}A^a\alpha^b(X),\quad P(X)=X-\sum_{ab} a_{ab}^{-1} A^a\alpha^b(X),
\end{align*}
where $a_{ab}^{-1}$ is the matrix inverse to the matrix $a_{ab}=\langle A^a,A^b\rangle$.

The affine connection is given by
\begin{align*}
\nabla^1_X Y &:=\nabla^g_X Y+\sum_{ab}a_{ab}^{-1} A^a(\nabla^g_X\alpha^b)(Y)\\
&=\nabla^g_X Y-\sum_{ab}a_{ab}^{-1} A^a\alpha^b(\nabla^g_X Y)
+\sum_{ab} A^a a_{ab}^{-1}X(\alpha^b(Y))\\
&=\nabla^g_X Y-Q(\nabla^g_X Y)+\sum_{ab} a_{ab}^{-1} X(\alpha^b(Y))A^a
\end{align*}

This connection is further studied by Bloch and Crouch \cite{BC}.

On the other side,  Lewis \cite{Lew} considered another natural affine connection:
\[
\nabla^2_X Y:=\nabla^g_X Y+(\nabla^g_X Q)(Y)=\nabla^g_X Y-Q(\nabla^g_X Y)+\nabla^g_X(Q(Y)).
\]

Note that although the definitions of $\nabla^1$ and $\nabla^2$ are quite similar, they are different:
\[
\nabla^2_X Y=\nabla^1_X Y+\sum_{ab} \alpha^b(Y)\nabla^g _X (a_{ab}^{-1}A^a)
\]

In both cases, they are extensions of $\nabla^P$:
\[
\nabla^1_X Y=\nabla^2_X Y=\nabla^g_X Y-Q(\nabla^g_X Y)=\nabla^P_X Y, \quad X\in\Gamma(TQ), \, Y\in\Gamma(\mathcal D),
\]
and their geodesic lines,
\begin{equation}\label{GL1}
\nabla^i _{\dot q} \,\dot q=0, \qquad i=1,2,
\end{equation}
with an initial condition $\dot q(t_0)\in\mathcal D_{q(t_0)}$ are also nonholonomic geodesic lines, i.e., the solutions of \eqref{GL}.

\begin{lem}\label{lema}
The equation \eqref{GL1} for $\nabla^1$ have a set of first integrals
\[
f^a=\alpha^a(\dot q)=\langle A^a,\dot q\rangle=c^a, \qquad a=1,\dots,k,
\]
which can be interpreted as affine nonholonomic constraints. In general, $f^a$ are not the first integrals of the geodesic flow of the connection $\nabla^2$.
\end{lem}

We note that in \cite{VeVe2, FeJo} the extended systems on the whole tangent spaces of the configuration spaces
of LR systems are considered (see \cite{Jo5, Jo6} for modified LR systems). They are the geodesic flows of the corresponding extended connections $\nabla^1$ with first integrals defining affine constraints.

\subsection{Partial connections}
Consider two natural mechanical nonholonomic systems  $(M,g,V,\mathcal D)$ and $(M,\tilde g,V,\mathcal D)$, such that
metrics $g$ and $\tilde g$ define the same orthogonal projections $P=\tilde P$ and that the restrictions of metrics
$g\vert_\mathcal D$ and $\tilde g\vert_\mathcal D$ coincides. Let $\nabla$ and $\tilde\nabla$ be the corresponding
Levi-Civita connections.

\begin{prop}\label{pomocna}
The nonholonomic systems $(M,g,V,\mathcal D)$ and $(M,\tilde g,V,\mathcal D)$ have the same solutions. However, their
vector bundle connections $\nabla^P$ and $\tilde{\nabla}^P$, in general, are different.
\end{prop}

\begin{proof}
For the Levi-Civita connection $\nabla$ we have the identity
\begin{align*}
g(\nabla_X\,Y,Z)=& \frac12\big(X g(Y,Z)+ Y g(X,Z)-Z g(X,Y) \big)\\
& +\frac12\big(g([X,Y],Z)- g([X,Z],Y)-g([Y,Z],X) \big),
\end{align*}
and the same expression for $\tilde g$ and $\tilde\nabla$. Therefore, if $X,Y\in\Gamma(\mathcal D)$, from $g_\mathcal D=\tilde g_\mathcal D$, $P=\tilde P$, we get
\[
g(\nabla_X\,Y,Z)=g(\tilde\nabla_X\, Y,Z) \qquad  \text{for all} \qquad Z\in\Gamma(\mathcal D),
\]
that is $P(\nabla_X\,Y)=P(\tilde\nabla_X\,Y)$, $X,Y\in\Gamma(\mathcal D)$.
Therefore they have the same nonholonomic equations \eqref{NHJ}.
On the other hand, for $X$ arbitrary and $Y,Z\in\Gamma(\mathcal D)$, we obtain
\[
g(\nabla_X\,Y-\tilde\nabla_X\, Y,Z)=g(\nabla_X\,Y,Z)-\tilde{g}(\tilde\nabla_X\, Y,Z)=\frac12\big(\tilde g(QX,Q[Y,Z])-g(QX,Q[Y,Z])\big).
\]
Thus, if the distribution $\mathcal D$ is integrable, the vector bundle connections $\nabla^P$ and $\tilde{\nabla}^P$ coincide, but in general they are different.
\end{proof}

Thus, the parallel transports of two vector bundle connections generally coincide only along the admissible paths. Therefore, the vector bundle connection $\nabla^P$ and the
corresponding extensions are not intrinsic objects related to the nonholonomic mechanical problems. On the other hand, from the proof of Proposition
\ref{pomocna}, we see that they have the same \emph{partial connection} (or \emph{nonholonomic connection}), defined as a restriction of $\nabla^P$ to the vector fields that are
sections of $\mathcal D$:
\[
\nabla^\mathcal{D}\colon \Gamma(\mathcal D)\times \Gamma(\mathcal D)\longrightarrow \Gamma(\mathcal D), \qquad \nabla^\mathcal{D}:=\nabla^P\vert_{\Gamma(\mathcal D)\times \Gamma(\mathcal D)}.
\]

Furthermore, the converse statements is also valid (see \cite{Go, DrGa}):
for a given a projection $P: TM\to \mathcal D$ and the Riemannian metric $g_\mathcal D$ on the distribution $\mathcal D$ there exist unique partial connection
$\nabla^\mathcal{D}\colon \Gamma(\mathcal D)\times \Gamma(\mathcal D)\longrightarrow \Gamma(\mathcal D)$,
that is a metric connection
\begin{equation}\label{a1}
Z\langle X,Y\rangle_\mathcal D=\langle \nabla^\mathcal{D}_Z X,Y\rangle_\mathcal D+\langle X,\nabla^\mathcal{D}_Z Y\rangle_\mathcal D, \qquad X,Y,Z\in\Gamma(\mathcal D)
\end{equation}
with vanishing torsion ${\bf T}$ of $\nabla^\mathcal{D}$ defined  by
\begin{equation}\label{a2}
{\bf T}(X,Y)=\nabla^\mathcal{D}_X\,Y - \nabla^\mathcal{D}_Y\, X -P[X,Y]=0, \qquad X,Y\in\Gamma(\mathcal D).
\end{equation}

Moreover, the projection of the gradient $\grad_\mathcal D V=P(\grad V)\in \Gamma(\mathcal D)$ can be defined in terms of the restriction of the metric to $\mathcal D$ as well:
\[
\langle \grad_\mathcal D V,X\rangle_\mathcal D=(dV,X) \qquad \text{for all} \qquad X\in\Gamma(\mathcal D).
\]

Therefore, the nonholomic equations \eqref{NHJ} are uniquely defined by the Riemannian metric $g_\mathcal D$ on $\mathcal D$,
the projection $P\colon TM\to \mathcal D$ (i.e, the transverse distribution $\mathcal D^\perp$), and the potential function $V$.
The classification of the $G$--invariant structures $(G,\mathcal D,g_\mathcal D, P)$  on 3--dimensional Lie groups $G$ is given in \cite{BBR}.

It seems that Vrancheanu (see \cite{Vr}) was the first who introduced a notion of partial connection.
In the same time, Synge \cite{Sy} considered connection on the vector bundle and its extension. The first notion of a curvature of the partial connection is given by Shouten \cite{Sh}, which is improved by Wagner \cite{DrGa, Go, Wa1}.
Note if the curvature of the vector bundle connection \eqref{nablap} is zero then we have $n-k$ independent vector fields parallel along arbitrary curves, whence, in particular, parallel transport \emph{along admissible paths} does not
depend of the path.
However, the converse is not true in general.
Wagner constructed an extension of the partial connection $\nabla^\mathcal D$ to a connection on the vector bundle $\mathcal D\to M$
such that its curvature is zero if and only if the parallel transport {along admissible paths} does not
depend of the path.

\section{Chaplygin reduction and connections}

\subsection{Chaplygin reduction}
Suppose that $\pi: M\to N=M/{G}$ is a principal bundle
with respect to the
{\it left\/} action of a Lie group ${G}$, and $\mathcal D$ is a principal
connection, i.e., $\mathcal D$ is a $G$-invariant distribution (collection of \emph{horizontal spaces}) and
$T_q M=\mathcal D_q\oplus \mathcal V_q$ for all $q$,
where $\mathcal V_q$ is tangent to the $G$--orbit through $q$ (the \emph{vertical space} at $q$).
Given a vector $X_q\in T_q M$, there is a decomposition $X_q=X_q^h+X_q^v$, where
$X_q^h\in \mathcal D_q$, $X_q^v\in \mathcal V_q$.
The {\it curvature} of $\mathcal D$ is the vertical valued 2-form $K$ on $M$
defined by
$$
K(X_q,Y_q)=-[\bar X_q^h,\bar Y_q^h]_q^v ,
$$
where $\bar X$ and $\bar Y$ are smooth vector fields on $M$ obtained by
extending of $X_q$ and $Y_q$.

In addition, suppose that
$G$ acts by isometries on $(M,g)$ and
that $V$ is $G$--invariant.
Then the equations \eqref{L-dAlembert} are $G$-invariant and the restriction $L\vert_\mathcal D$
induces the reduced Lagrangian $L_{red}$, i.e, the reduced metric $g_0$ and the reduced potential energy $V_0$, via identification
$TN\approx \mathcal D/{G}$.
The  {\it reduced Lagrange--d'Alambert} equations on the tangent bundle $TN$ take the form
\begin{equation}
\big(\frac{\partial L_{red}}{\partial x}-\frac{d}{dt}\frac{\partial L_{red}}{\partial
\dot x},\eta\big)= \langle\dot x^h,K_q(\dot x^h,\eta^h)\rangle\vert_q \quad
\mathrm{for\;all} \quad \eta\in T_x N,
\label{ChaplyginRed}
\end{equation}
where $q\in\pi^{-1}(x)$ and $\dot x^h$ and $\eta^h$ are unique horizontal lifts of $\dot x$ and $\eta$ at $q$. The right-hand side of \eqref{ChaplyginRed} can be written as $\Sigma(\dot x,\dot x,\eta)$ where
$\Sigma$ is $(0,3)$--tensor field on the base manifold $N$ defined by
\begin{equation}\label{polje}
\Sigma_x(X,Y,Z)=\langle X^h,K_q(Y^h,Z^h)\rangle\vert_q, \qquad q\in\pi^{-1}(x).
\end{equation}

The system $(M,g,V,\mathcal D,G)$ is referred to as {\it a $G$--Chaplygin system} \cite{Koi}, as a generalization of classical Chaplygin systems
with Abelian symmetries \cite{Ch2}.

\subsection{Reduced connections}
Let $\nabla^0$ be
the Levi-Civita connection of the reduced metric $g_0$.
Then, by the use of \eqref{veza}, the reduced equation \eqref{ChaplyginRed} can be written in the form
\[
\langle\nabla^0_{\dot x} \, \dot x+\grad_0 V_0(x),\eta\rangle_0+\Sigma(\dot x,\dot x,\eta)=0,
\]
that is,
\begin{equation}\label{CaplJedn}
\nabla^0_{\dot x}\,\dot{x}+B(\dot x,\dot x)=-\grad_0 V_0(x),
\end{equation}
where the gradient is taken with respect to $g_0$ and
$(1,2)$--tensor field $B$ is defined by
\[
\langle B(X,Y),Z\rangle_0=\Sigma(X,Y,Z)
\]
("raising up of the third index" in $\Sigma$).
Now, the equations \eqref{CaplJedn} can be written as
\[
\nabla^{B}_{\dot x}\,\dot x=-\grad_0 V_0,
\]
where $\nabla^B$ is a symmetric connection defined by
\[
\nabla^B_X\,Y=\nabla^0_X\,Y+\frac12\big(B(X,Y)+B(Y,X)\big).
\]
It is clear that $\nabla^B$ is metric, i.e, Levi-Civita connection, if and only if $B$ is skew--symmetric.
To the authors knowledge, the connection $\nabla^B$, for Abelian Chaplygin systems, is firstly introduced by Bak\v sa in \cite{Ba}.

On the other hand, in the case of $G$--Chaplygin systems,  $\nabla^P$ is $G$--invariant and descends to the reduced connection
$\tilde\nabla_X\, Y=\pi_*(\nabla^P_{X^h}\,Y^h)$ (see \cite{Koi, Co}),
and the solutions of \eqref{NHJ},
project onto the solution of
$$
\tilde\nabla_{\dot x}\,\dot x=-\grad_0 V_0.
$$

The connection $\tilde\nabla$ is given by the following expression (see Koiller \cite{Koi} and Cortes \cite{Co}):
\[
\tilde\nabla_X\, Y=\nabla^0_X\,Y+\frac12\big(B(X,Y)+B(Y,X)-C(X,Y)\big),
\]
where $C$ is a $(1,2)$--tensor field
\[
\langle X,C(Y,Z)\rangle_0=\Sigma(X,Y,Z)
\]
("raising up of the first index" in $\Sigma$).
From the skew-symmetry of the curvature, we have that $C$ is skew-symmetric as well: $C(X,Y)=-C(Y,X)$.

Note that, since $\nabla^P$ is metric, $\tilde\nabla$ is a metric connection as well.
It is symmetric (and whence Levi-Civita) if and only if $C\equiv 0$, which is the same as $B\equiv 0$, i.e., $\nabla^0=\tilde\nabla=\nabla^B$.

The torsion of $\tilde\nabla$ is projection of the torsion \eqref{torsion}:
\begin{equation*}
\tilde T(X,Y)=\pi_* T(X^h,Y^h)=-\pi_*Q[X^h,Y^h],
\end{equation*}
which equals to zero if $Q[X^h,Y^h]=0$ or $\mathcal D^\perp=\mathcal V$ (here $X^h$ and $Y^h$ are horizontal lifts of $X$ and $Y$). The first condition is equivalent to the integrablity of the distribution, while the second condition is equivalent to the vanishing of reaction forces.

\section{Time reparametrisation and conformal metrics}\label{conformal}

\subsection{Variation of the Chaplygin multiplier method}
The nonholonomic equations are not Hamiltonian. For the reduced Abelian Chaplygin systems \eqref{ChaplyginRed}, Chaplygin proposed
the Hamiltanization method using a time reparametrization $d\tau=\nu(x)dt$
now referred as a Chaplygin multiplier (see \cite{Ch2}). Geometrically, this means that the reduced system is conformally Hamiltonian.
Various equivalent ways of Hamiltonization of $G$--Chaplygin systems
and the relationship with an existence of an invariant measure can be found in \cite{Co, CCLM, EKR, FeJo, St}.
Let us point out that the existence of an invariant measure (e.g., see \cite{BiBoMa, BoKaPi, BoKu, FNM, RaPr, ZB}), whence a possible Hamiltonization, is not typical for nonholonomic systems.

In terms of connections, the Chaplygin multiplier
is a function $\nu(x)\ne 0$ such that the reduced equation \eqref{CaplJedn} in the new time takes the form
\begin{equation}\label{CCE}
\nabla^*_{x'}\,x'=-\grad_{*} V_0,
\end{equation}
where $\nabla^{*}$ is the Levi-Civita connection of the conformal metric  $g_{*}=\nu^2 g_0$ on the base manifold $N$. The Chaplygin time reparametrization in a framework of connections is also studied in \cite{FB}.

We will slightly modify the Chaplygin method by allowing the conformal factor and multiplier $\nu$
to be independent. Let us consider a general setting, where conformal metrics $g_*=f^2 g$ and $g$ are given on a manifold $M$ with local coordinates $q=(q^1,\dots,q^n)$ ($f\ne 0$ on $M$).
The coefficients of their Levi-Civita connections $\nabla^*$ and $\nabla$ are related
by
\begin{equation}\label{g*}
\Gamma^{*k}_{ij}=\Gamma^k_{ij}+\frac{1}{f}\big(\delta^{k}_j\frac{\partial f}{\partial q^i}+\delta^k_i \frac{\partial f}{\partial q^j}-g_{ij}g^{kl}\frac{\partial f}{\partial q^l}  \big).
\end{equation}

Consider the geodesic equations  on $(M,g_*)$,
\begin{equation}\label{gprime}
\frac{ d^2 q^k}{d\tau^2}+\Gamma^{*k}_{ij}\frac{dq^i}{d\tau}\frac{dq^j}{d\tau}=0,
\end{equation}
with respect to an affine parameter $\tau$.
The left hand side of  \eqref{gprime},
after a time-reparametrisation
\begin{equation}\label{NewTime}
d\tau=\nu(q)dt\colon \quad \dot q=\nu\cdot q',
\end{equation}
gets a well known expression
\begin{equation}\label{gdot}
\frac{ d^2 q^k}{d\tau^2}+\Gamma^{*k}_{ij}\frac{dq^i}{d\tau}\frac{dq^j}{d\tau}=\frac{1}{\nu^2}\big(\ddot q^k-\frac{\dot{\nu}}{\nu}\dot q^k+\Gamma^{*k}_{ij}\dot q^i\dot q^j\big).
\end{equation}

Now, combining  \eqref{g*} and \eqref{gdot}, we obtain the identity
\begin{equation}\label{gdot*}
\frac{ d^2 q^k}{d\tau^2}+\Gamma^{*k}_{ij}\frac{dq^i}{d\tau}\frac{dq^j}{d\tau}=\frac{1}{\nu^2}\big(\ddot q^k+\Gamma^{k}_{ij}\dot q^i\dot q^j-\frac{\dot{\nu}}{\nu}\dot q^k+\frac{1}{f}\big(\delta^{k}_j\frac{\partial f}{\partial q^i}+\delta^k_i \frac{\partial f}{\partial q^j}-g_{ij}g^{kl}\frac{\partial f}{\partial q^l}  \big)\dot q^i\dot q^j\big).
\end{equation}

Thus, the geodesic equations \eqref{gprime} in the time $t$ take the form
\begin{equation}\label{fN}
\ddot q^k+\Gamma^{k}_{ij}\dot q^i\dot q^j=\frac{\partial \ln\nu}{\partial q^r}\dot q^r\dot q^k-\frac{1}{f}\big(2\frac{\partial f}{\partial q^i}\dot q^i\dot q^k-g_{ij}g^{kl}\frac{\partial f}{\partial q^l}\dot q^i\dot q^j  \big).
\end{equation}

Let us define a $\dot q$--dependent vector field on $M$:
\begin{equation}\label{dobarF}
F=\langle {\grad}\ln \nu,\dot q\rangle\dot q-2\langle {\grad}\ln f,\dot q\rangle\dot q+\langle \dot q,\dot q\rangle {\grad} \ln f.
\end{equation}
Then we can write the identity \eqref{gdot*} and system \eqref{fN} in an invariant form
\begin{equation}\label{invarijantnaForma}
\nabla^*_{q'}\, q'=\nu^{-2}\big(\nabla_{\dot q}\,\dot q-F(\dot q,q)\big)
\end{equation}
and $\nabla_{\dot q}\,{\dot q}=F(\dot q,q)$, respectively.

\begin{prop}\label{stav}
Consider the Newton type equation
\begin{equation}\label{Newton}
\nabla_{\dot q}\,{\dot q}=F(\dot q,q)
\end{equation}
on a Riemannian manifold $(M,g)$,
such that the force field can be written in the form \eqref{dobarF}
for certain functions $f,\nu$ different from zero on $M$.
Then, after a time reparametrisation \eqref{NewTime}, the equation takes the form of the geodesics equation
\begin{equation}\label{Geq}
\nabla^*_{q'}\,q'=0
\end{equation}
of the metric $g_*=f^2 g$.
\end{prop}

Assuming $\nu=f^\alpha$, the expression \eqref{dobarF} is slightly simplified:
\begin{equation}\label{force}
F=(\alpha-2)\langle {\grad}\ln f,\dot q\rangle\dot q+\langle \dot q,\dot q\rangle {\grad}\ln f.
\end{equation}

Taking $\alpha=1$ and $(M,g)$ to be the reduced space $(N,g_0)$ of the $G$-Chaplygin system,
the force field \eqref{force} reads
\[
F=-\langle {\grad_0}\ln\nu,\dot x\rangle_0\dot x+\langle \dot x,\dot x\rangle_0 {\grad_0}\ln\nu,
\]
which is exactly the term $-B(\dot x,\dot x)$ in \eqref{CaplJedn} when the method of the Chaplygin multiplier is applicable (e.g., see \cite{FeJo}). This is also the case of generalized Chaplygin method  (see \cite{BBM}). It remains to note that Proposition \ref{stav} is without a potential force field.
But for $\nu=f$,
the gradients $\grad_0 V_0$ and $\grad_* V_0$ of the metrics $g_0$ and $g_*=\nu ^{2} g_0$ are related by
$\grad_* V_0=\nu^{-2}\grad_0 V_0$.
Thus, from \eqref{invarijantnaForma}, we have the identity
\[
\nabla^*_{x'}\,x'+\grad_*V_0=\nu^{-2}\big(\nabla^0_{\dot x}\,\dot x+B(\dot x,\dot x)+\grad_0 V_0(x)\big)
\]
and \eqref{CaplJedn} takes the form \eqref{CCE} (the Chaplygin multiplier does not depend on the potential).

Note that the geodesic equation \eqref{Geq} has the kinetic energy integral $\frac12\langle q',q'\rangle_*$. Therefore,  the system
\eqref{Newton},  \eqref{dobarF} has the quadratic first integral
${f^2}/{2\nu^2} \langle\dot q,\dot q\rangle$, which is an obstruction to the construction.
In particular, when the force $F(\dot q,q)$ is gyroscopic, that is $\langle F(\dot q,q),\dot q\rangle=0$, the system preserves the kinetic
energy $\frac12\langle\dot q,\dot q\rangle$ as is the case of $G$--Chaplygin systems. Whence,
 $f^2/\nu^2$ is a constant and we essentially have the Chaplygin construction with $f=\nu$.

However, Proposition \ref{stav} can be formulated also with a weaker assumption:
for the Newton equation \eqref{Newton}
having an invariant relation
\begin{equation}\label{invariantSurface}
\mathcal E=\Big\{(\dot q,q)\in TM\,\vert\, \frac12 \langle\dot q,\dot q\rangle-\frac{\nu^2}{f^2}=0\Big\},
\end{equation}
when the force $F$ restricted to $\mathcal E$ reads
\begin{align}\label{force2}
F=&\langle {\grad}\ln \nu,\dot q\rangle\dot q-2\langle {\grad}\ln f,\dot q\rangle\dot q+
2\frac{\nu^2}{f^2} \grad \ln f  \\
\nonumber &(=(\alpha-2)\langle {\grad}\ln f,\dot q\rangle\dot q+f^{2\alpha-4}\grad f^2, \quad \text{for}\quad \nu=f^\alpha).
\end{align}
Then the solution of \eqref{Newton}, \eqref{force2} that belong to the invariant surface \eqref{invariantSurface} are mapped
to the geodesic lines \eqref{Geq} with the unit kinetic energy $\frac12\langle q',q'\rangle_*=1$.

\begin{remark}\label{AA}
{\rm
The transformations between natural mechanical systems, that include time rescaling, are already studied by Painlev\'e and  Levi-Civita.
In a general form the transformations for time-independent and time-dependent systems are
described by Thomas \cite{Th} and Lichnerowitz and Aufenkamp \cite{LA}, respectively. For a recent results, see \cite{Al}.
The main difference
with respect to Chaplygin's method as presented here, is that the force fields considered in \cite{Al, LA, Th} do not depend on velocities.
}
\end{remark}

\subsection{The nonholonomic Maupertuis principle}
In the case $\alpha=2$, \eqref{force2} reads $F=\grad f^2$.
Taking $f=\sqrt{h-V(q)}$, the invariant relation \eqref{invariantSurface} and the force field become
\begin{equation}\label{isoenergetic}
\mathcal E_h: \quad \frac12 \langle\dot q,\dot q\rangle+V(q)=h
\end{equation}
and
$F=-\grad V$, respectively.
In other words,
we obtain a well known formulation of the \emph{Maupertuis principle}: the solutions $q(t)$ of
the Newton equation
$\nabla_{\dot q}\,\dot q=-\grad{V}$
with the fixed energy \eqref{isoenergetic}, in the new time
\begin{equation}\label{NT}
d\tau=(h-V)dt,
\end{equation}
are the geodesic lines $q(\tau)$
of the \emph{Jacobi metric}
$g_J=(h-V)g$
with the unit kinetic energy $\frac12 \langle q',q'\rangle_J=1$.

Further, for the curves $q(t)$ within the region $V(q)<h$ that belong to the isoenergetic surface \eqref{isoenergetic}, the identity
\eqref{invarijantnaForma} reads
\begin{equation}\label{MJ}
\nabla^J_{q'}\,q'=(h-V)^{-2}\big(\nabla_{\dot q}\,\dot q+\grad V\big),
\end{equation}
where $\nabla^J$ is the Levi-Civita connection of the {Jacobi metric}. Substituting \eqref{MJ} to the nonholonomic equation \eqref{NJ}, we obtain
\begin{equation*}\label{NJnova}
\langle \nabla_{\dot q} \,\dot q+\mathrm{grad}\, V(q),\xi\rangle =(h-V)^2\langle \nabla^J_{q'}\,q',\xi \rangle=(h-V)\langle \nabla^J_{q'}\,q',\xi \rangle_J=0,
\end{equation*}
for all $\xi\in\mathcal D_q$.
Whence, we get
the \emph{nonholonomic Maupertuis principle}: after changing of time \eqref{NT}, the solutions $q(t)$ of \eqref{NHJ}
that belong to the isoenergetic surface \eqref{isoenergetic} become the nonholomic geodesic lines
\begin{equation}\label{GLMJ}
\nabla^{JP} _{q'} \,q'=0, \qquad q'\in\mathcal D_q,
\end{equation}
with the unit kinetic energy $\frac12 \langle q',q'\rangle_J=1$,
where
\[
\nabla^{JP}_X Y=P(\nabla^J_X Y), \quad X\in\Gamma(TM), \quad Y\in\Gamma(\mathcal D).
\]

Note, since $g$ and $g_J$ are conformal, the orthogonal projection $P$ is the same for both metrics.
The above version of the Maupertuis principle can be found in Koiller \cite{Koi} (see also Synge \cite{Sy}).

For $G$--Chaplygin systems, as above,
let $\nabla^0$ and $\nabla^{J}$ be
the Levi-Civita connections of the reduced metric $g_0$ and of the Jacobi metric $g_{J}=(h-V_0)g_0$, and consider
a symmetric connection
\begin{equation}\label{baksa}
\nabla^{JB} _{X} \,Y=\nabla^{J}_X\, Y+\frac12\big(B(X,Y)+B(Y,X)\big)
\end{equation}
on the base manifold $N=M/G$. After a time rescaling $d\tau=(h-V_0)dt$, we have
the identity
\begin{equation}\label{MJCJ}
\nabla^{J}_{x'}\,x'+B(x',x')=(h-V_0)^{-2}\big(\nabla^0_{\dot x}\,\dot x+B(\dot x,\dot x)+\grad_0 V_0\big),
\end{equation}
which implies the following variant of the Maupertuis principle derived by Bak\v sa for Abelian systems \cite{Ba}:
the solution of the reduced equation \eqref{CaplJedn}
that satisfy
$\frac12 \langle \dot x,\dot x\rangle_0+V_0(q)=h$
in the new time $\tau$ are the geodesic lines
\[
\nabla^{JB} _{x'} \,x'=0
\]
of the connection \eqref{baksa} with the unit kinetic energy $\frac12 \langle x',x'\rangle_{J}=1$.

\section{Integrability of the Chaplygin ball rolling over a sphere in $\R^n$}

\subsection{Definition of the system}
We consider the Chaplygin ball type problem of rolling without
slipping and twisting of an $n$-dimensional balanced ball of radius $\rho$ (the mass center $C$
coincides with the geometrical center) in the following
nonholonomic problems (for more details, see \cite{Jo6}):

\begin{itemize}
\item[(i)] rolling over outer surface of the $(n-1)$-dimensional
fixed sphere of radius $\sigma$;

\item[(ii)] rolling over inner surface of the $(n-1)$-dimensional
fixed sphere of radius $\sigma$ ($\sigma>\rho$);

\item[(iii)] rolling over outer surface of the $(n-1)$-dimensional
fixed sphere  of radius $\sigma$, but the fixed sphere is within
the rolling ball ($\sigma<\rho$, in this case, the rolling ball is
actually a spherical shell).
\end{itemize}

For the configuration space we take the {direct product} of the Lie
group $SO(n)$ and the sphere $\{\mathbf r\in\R^n\,\vert\,(\mathbf r,\mathbf r)=(\sigma\pm\rho)^2\}$,
where we take "$+$"  for the case (i) and "$-$"  for the cases
(ii) and (iii).
The matrix $g\in SO(n)$ maps a frame attached to the body to the
space frame and  $\mathbf r=\overrightarrow{OC}$
is the position vector of the ball center $C$ in the space frame,
and the origin $O$ coincides with the center of the fixed sphere.

\begin{pspicture}(12,7)
\pscircle[linecolor=black](3.5,3){2}
\psellipticarc[linestyle=dashed](3.5,3)(2,0.8){0}{180}
\psellipticarc(3.5,3)(2,0.8){180}{360}
\psdot[dotsize=2pt](3.5,3)\uput[0](3.5,3.2){$O$}
\psdot[dotsize=2pt](1.7,5.4)\uput[0](1.7,5.6){$C$}
\pscircle[linecolor=black](1.7,5.4){1.03}
\psarc[linewidth=0.03cm](1.7,5.4){0.92}{200}{250}
\psarc[linewidth=0.03cm](1.7,5.4){0.82}{215}{245}
\psline{<-}(1.91,5.1)(2.33,4.6)
\uput[0](1.9,5.1){$\gamma$}
\pscircle[linecolor=black](9,3){2}
\psellipticarc[linestyle=dashed](9,3)(2,0.8){0}{105}
\psellipticarc[linestyle=dashed](9,3)(2,0.8){165}{180}
\psellipticarc(9,3)(2,0.8){180}{360}
\psdot[dotsize=2pt](9,3)\uput[0](9,3.2){$O$}
\psdot[dotsize=2pt](8.0,3.8)\uput[0](8.0,3.8){$C$}
\pscircle[linecolor=black](8.0,3.8){0.703}
\psarc[linewidth=0.03cm](8.0,3.8){0.61}{200}{250}
\psarc[linewidth=0.03cm](8.0,3.8){0.55}{215}{245}
\psline{<-}(7.02,4.63)(7.46,4.27)
\uput[0](7.2,4.7){$\gamma$}
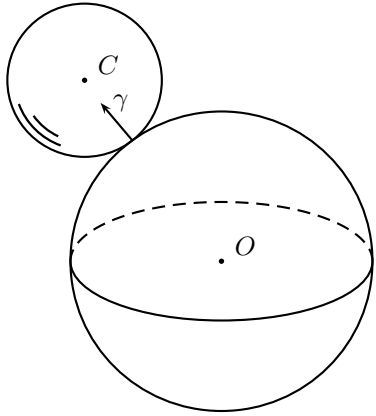
\uput[0](0.8,0.5){Fig 1a: Rolling over a sphere}
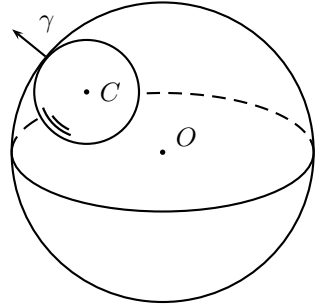
\uput[0](6.4,0.5){Fig 1b: Rolling within a sphere}
\end{pspicture}

The condition that the ball is rolling without
slipping and twisting at the contact point determines $\frac12n(n-1)$ constraints in velocities $(\dot g,\dot{\mathbf r})$, and the
corresponding $(n-1)$-dimensional constraint distribution $\mathcal D$ defines a principal connection of the bundle
\begin{equation}\label{principal}
SO(n) \longrightarrow   SO(n) \times S^{n-1} \overset{\pi}\longrightarrow  S^{n-1}
\end{equation}
with respect to the diagonal \emph{left} $SO(n)$-action:
$a\cdot (g,\mathbf r)=(ag,a\mathbf r)$, $a\in SO(n)$  (see \cite{Jo6}).
Here the submersion $\pi$ is given by
\begin{equation}\label{submersion}
\gamma=\pi(g,\mathbf
r)=\frac{1}{\sigma\pm\rho}g^{-1}\mathbf r
\end{equation}
and $\gamma$  is a unit vector, the direction of the contact point in the frame attached to the ball. Thus, the
problem of the rubber rolling of a ball over a fixed sphere
is a $SO(n)$--Chaplygin system and
 reduces to the tangent bundle $TS^{n-1}\cong \mathcal
D/SO(n)$.

Let $\mathbb I: so(n) \to so(n)$ be the inertia
tensor, $m$ mass of the ball, and $(\cdot,\cdot)$ the Euclidean scalar product in $\R^n$.
The reduced
Lagrangian, the reduced metric $g_0$, and
the $(0,3)$--tensor field \eqref{polje}
 on $S^{n-1}$ read (see \cite{Jo6})
\begin{align*}
& L_{red}(\dot\gamma,\gamma) =-\frac{1}{4\epsilon^2}\tr(\mathbf
I(\gamma\wedge\dot\gamma)\circ(\gamma\wedge\dot\gamma) )=-\frac{1}{2\epsilon^2}(\mathbf
I(\gamma\wedge\dot\gamma) \gamma,\dot\gamma ),\\
& \langle X,Y\rangle_0=-\frac{1}{2\epsilon^2}\tr( \mathbf
I(\gamma\wedge X)\circ(\gamma\wedge Y))=-\frac{1}{\epsilon^2}( \mathbf
I(\gamma\wedge X)\gamma,Y),\\
&\Sigma_\gamma(X,Y,Z)=\frac{2\epsilon-1}{2\epsilon^3}\tr(\mathbf
I(\gamma\wedge X)\circ (Y\wedge Z))=\frac{2\epsilon-1}{\epsilon^3}(\mathbf
I(\gamma\wedge X)Y,Z),
\end{align*}
where $\mathbf I=\mathbb I+D\cdot \mathrm{Id}_{so(n)}$, $D=m\rho^2$,
and $\epsilon={\sigma}/({\sigma\pm\rho})$. Note that when radii of sphere and the ball are equal ($\epsilon=\frac12$), $\Sigma_\gamma\equiv0$
and the reduced system is Hamiltonian without a time reparametrization (for $n=3$ see \cite{EK, BMT}).

Thus,  the  equation describing the motion of the
reduced system \eqref{ChaplyginRed} are
\begin{equation}\label{REDUCED}
\Big(\epsilon\frac{d}{dt}\big(\mathbf
I(\gamma\wedge\dot\gamma)\gamma\big)+(1-\epsilon)\mathbf
I(\gamma\wedge\dot\gamma)\dot\gamma,\xi\Big) =0 , \qquad \xi\in
T_\gamma S^{n-1}.
\end{equation}

The system always has an invariant measure (see \cite{Jo5, Jo6}). The density of a measure significantly simplifies for
a special inertia operator
\begin{align}
\label{spec-op} \mathbb I(E_i\wedge E_j)=(a_ia_j-D)E_i\wedge
E_j \quad \text{i.e.}, \quad  \mathbf I(X \wedge Y)=AX\wedge
AY,
\end{align}
where $A=\diag(a_1,\dots,a_n)$.
Then the reduced Lagrangian $L_{red}$, the reduced metric $g_0$, and the equation \eqref{REDUCED} take the form
\begin{align}
\label{LRED} & L_{red} = \frac 1{2\epsilon^2} \left((A\dot \gamma,\dot
\gamma)(A\gamma,\gamma)- (A\gamma,\dot \gamma)^2 \right),\\
&\label{REDM0} \langle X,Y\rangle_0=\frac{1}{\epsilon^2}\big((A X, Y)(A\gamma,\gamma) -(A\gamma, X)(A\gamma,Y)\big),
\quad X,Y\in T_\gamma S^{n-1},\\
& \label{REDUCED2}\epsilon\frac{d}{dt}
\big( (A\gamma,\dot\gamma) A\gamma -
(A\gamma,\gamma) A \dot\gamma\big)=(\epsilon-1)\big( (A\dot\gamma,\dot\gamma) A\gamma -
(A\gamma,\dot\gamma) A \dot\gamma\big)+\lambda\gamma,
\end{align}
respectively. Here, $\lambda$ is the Lagrange multiplier determined by the condition that a trajectory $\gamma(t)$ belongs to the sphere
$\langle \gamma,\gamma\rangle=1$.

\begin{thm}{\rm \cite{Jo6}}\label{REDsym}
Under a time substitution  $d\tau
=\nu(\gamma)=\epsilon(A\gamma,\gamma)^{\frac{1}{2\epsilon}-1}\, dt$, the reduced system
\eqref{REDUCED2} becomes
the geodesic flow of the metric $g_*=\nu^2 g_0$:
\begin{equation}\label{th51}
\langle X,Y\rangle_*=(\gamma, A\gamma)^{\frac{1}{\epsilon}-2}
\big(\big(A X,Y)(A\gamma,\gamma) -(A\gamma,X)(A\gamma,Y)\big),\quad X,Y\in T_\gamma S^{n-1}.
\end{equation}
\end{thm}

The procedure of reduction for rubber rolling over a sphere for
$n=3$ is given by Ehlers and Koiller \cite{EK}. In this
case the system is always Hamintonizable due to the fact that it
has an invariant measure and that the reduced configuration space
is 2--dimensional. In 3-dimensional case all inertia operators can be written in the form \eqref{spec-op}
and Theorem \ref{REDsym} reduces to the one
given in \cite{EK}.

\subsection{Integrability for
$\rho=2\sigma$}

Remarkably, for $n=3$, Borisov and Mamaev proved the
integrability of the rubber rolling for a specific ratio between radiuses of the ball and the spherical shell
(the case (iii), where $\rho=2\sigma$, i.e, $\epsilon=-1$), see \cite{BM2}.
We proceed in proving the complete integrability of the $n$-dimensional variant of the problem.

\begin{lem}\label{lemica}
Under the mapping
\begin{equation}
\label{transf} x=\frac{A^\frac12\gamma}{\sqrt{(A\gamma,\gamma)}}.
\end{equation}
the metric \eqref{th51} transforms to the metric
\begin{equation}\label{novaG}
\mathbf g(X,Y)=(x, A^{-1}x)^{-\frac{1}{\epsilon}}{(X,Y)}, \qquad X,Y\in T_{x} S^{n-1},
\end{equation}
conformally equivalent to the standard metric on the sphere
\begin{equation}
(x,x)=1.  \label{x1}
\end{equation}
\end{lem}

Considered a natural mechanical system on the sphere \eqref{x1}
with the Lagrangian
\begin{align*}
L_\epsilon=\frac12\big(\frac{dx}{ds},\frac{dx}{ds}\big)-V_{\epsilon}(x), \quad V_{\epsilon}(x)=-(A^{-1}x,x)^{-\frac{1}{\epsilon}}.
\end{align*}

\begin{prop}\label{jakobi}
The trajectories $\gamma(t)$ of the rolling of a rubber Chaplygin ball over
a spherical surface determined by equation \eqref{REDUCED2}, with the unit kinetic energy $\frac12\langle \dot\gamma,\dot\gamma\rangle_0=1$, under the
transformation \eqref{transf} and
a time reparametrisation
\begin{align*}
ds=\epsilon (A^{-1} x,x)^{1+\frac{1}{2\epsilon}}dt\quad(=\epsilon (A\gamma,\gamma)^{-1-\frac{1}{2\epsilon}}dt),
\end{align*}
are mapped to the zero-energy trajectories $x(s)$  of the natural
mechanical systems with the Lagrangian $L_\epsilon$:
\begin{equation}\label{prirodni}
\frac{d^2}{ds^2}x=-\frac{2}{\epsilon}\big(A^{-1}x,x\big)^{-\frac{1}{\epsilon}-1}A^{-1}x+\lambda x, \quad
\lambda=\frac{2}{\epsilon}\big(A^{-1}x,x\big)^{-\frac{1}{\epsilon}}-\big(\frac{dx}{ds},\frac{dx}{ds}\big).
\end{equation}
\end{prop}

\begin{proof}
According to the {Maupertuis principle}, the trajectories $x(s)$ of the system with Lagrangian $L_\epsilon$ laying on the zero-energy invariant surface
\begin{equation}\label{IS}
\frac12(\frac{dx}{ds},\frac{dx}{ds})-(A^{-1}x,x)^{-\frac{1}{\epsilon}}=0,
\end{equation}
after a time
reparametrization
\[
d\tau=(A^{-1}x,x)^{-\frac{1}{\epsilon}}ds,
\]
become the geodesic lines $x(\tau)$ of \eqref{novaG} with the unit kinetic energy $\frac12 \mathbf g (x',x')=1$ ($x'=dx/d\tau$). On the other hand,
from  Theorem \ref{REDsym} and Lemma \ref{lemica},
the solutions $\gamma(t)$ of the equation \eqref{REDUCED2}, after a time reparametrization
\[
d\tau=\epsilon(A\gamma,\gamma)^{\frac{1}{2\epsilon}-1}\, dt=\epsilon (A^{-1} x,x)^{1-\frac{1}{2\epsilon}}dt
\]
become the geodesic lines $x(\tau)$ of the metric \eqref{novaG} with the same kinetic energy
\[
\frac12\langle \dot\gamma,\dot\gamma\rangle_0=\frac12\langle \gamma',\gamma'\rangle_*=\frac12\mathbf g(x',x')
\]
(see the discussion after Proposition \ref{stav}).
Combining the above transformations we obtain the proof of the statement.
\end{proof}

\begin{pspicture}(12,6)
\pscircle[linecolor=black](5.2,4.1){1}
\psellipticarc[linestyle=dashed](5.2,4.1)(1,0.4){0}{180}
\psellipticarc(5.2,4.1)(1,0.4){180}{360}
\psdot[dotsize=2pt](5.2,4.1)\uput[0](4.5,4.1){$O$}
\psdot[dotsize=2pt](6.0,3.5)\uput[0](6.0,3.5){$C$}
\pscircle[linecolor=black](6,3.5){2}
\psarc[linewidth=0.03cm](6,3.5){1.88}{200}{250}
\psarc[linewidth=0.03cm](6,3.5){1.80}{215}{245}
\psline{->}(4.38,4.68)(3.57,5.29)
\uput[0](3.9,5.1){$\gamma$}
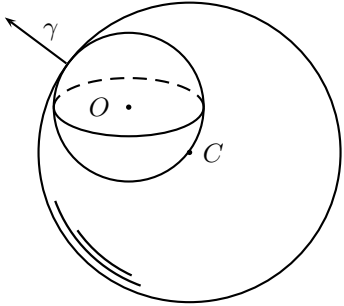
\uput[0](0,0.6){Fig 2: Rolling shell over a fixed sphere placed inside: integrable case $\epsilon=-1$}
\end{pspicture}

\begin{remark}
Proposition \ref{jakobi} can be seen as a variant of the construction given in Section \ref{conformal} as well,
where we take
$\nu(\gamma)=\epsilon (A\gamma,\gamma)^{-1-\frac{1}{2\epsilon}}$ and
$f(\gamma)=\epsilon (A\gamma,\gamma)^{-1}$.
\end{remark}

Among the potentials $V_{\epsilon}$, there are two
exceptional ones determining completely integrable systems: for $\epsilon=+1$ we have Braden's \cite{Br} and
for $\epsilon=-1$  Neumann's potential \cite{Moser}.

\begin{thm}\label{glavna}
For an inertia operator \eqref{spec-op} and $\rho=2\sigma$ $(\epsilon=-1)$, the reduced problem of the rolling of a rubber Chaplygin ball over
a spherical surface is completely integrable:
under the
transformation \eqref{transf} and a
time reparametrisation
\begin{align*}
ds=- (A^{-1} x,x)^{\frac{1}{2}}dt\quad(=- (A\gamma,\gamma)^{-\frac{1}{2}}dt),
\end{align*}
the solutions $\gamma(t)$ of \eqref{REDUCED2} with the unit kinetic energy $\frac12\langle \dot\gamma,\dot\gamma\rangle_0=1$
are mapped to the zero-energy trajectories $x(s)$  of the Neumann system with Lagrangian $L_{-1}$.
\end{thm}

If the radius $\sigma$ of the fixed sphere  tends to infinity, the parameter $\epsilon$ tends to 1, and the system transforms to the rolling of the rubber Chaplygin ball over a horizontal hyperplane in $\R^n$ \cite{Jo3} (for $n=3$, see \cite{EKR}).
The equation \eqref{REDUCED2} for $\epsilon=1$,
\begin{equation}\label{REDUCED3}
\frac{d}{dt}
\big( (A\gamma,\dot\gamma) A\gamma -
(A\gamma,\gamma) A \dot\gamma\big)=\lambda\gamma,
\end{equation}
coincides with the reduced equation of the nonholonomic Veselova problem studied in \cite{FeJo}.
In the 3-dimensional case the paper \cite{VeVe2} established a  relation between the Veselova
problem and the Neumann system.
This result was generalized in \cite{FeJo} as follows.
Under a time reparameterization
\begin{equation*}
ds_1=\sqrt{\frac{(A\dot\gamma,\dot \gamma)(A\gamma,\gamma)- (A\gamma,\dot \gamma)^2 } {(A\gamma,\gamma)}}\, dt  \label{4.6}
\end{equation*}
the solutions $x(t)=\gamma(t)$ of \eqref{REDUCED3}
 transform to solutions of the Neumann problem
on the sphere \eqref{x1} with the potential $V(x)=\frac12 (A^{-1} x,x)$,
\[
\frac{d^2}{ds_1^2}x=-A^{-1}x+\lambda x, \quad \lambda=\big(A^{-1}x,x\big)-\big(\frac{dx}{ds_1},\frac{dx}{ds_1}\big),
\]
that belong to the invariant set
$\big( A\frac{dx}{ds_1} ,\frac{dx}{ds_1} \big) (Ax,x)-
\big( Ax, \frac{dx}{ds_1} \big)^2-(A x,x)=0$.
From Proposition \ref{jakobi} we get another trajectory isomorphism.

\begin{thm}\label{veselovaRevised}
Under the
transformation \eqref{transf} and
a time reparametrization
$ds=(A^{-1} x,x)^{\frac{3}{2}}dt$,
the unit kinetic energy trajectories $\gamma(t)$ of \eqref{REDUCED3}
are mapped to the zero-energy trajectories $x(s)$ of the Braden system with the Lagrangian $L_{+1}$.
\end{thm}

\subsection{Separation of variables in the case $\rho=2\sigma$}\label{razdvajanje}

Here we assume $a_1>a_2>...>a_n>0$ and $\epsilon=-1$.
In the three-dimensional case Borisov and Mamaev constructed separating variables of the system as a deformation of sphero-conical coordinates \cite{BM2}. A similar type of deformations are used in \cite{BMF, Ts, Ts2}, where they are called nonholonomic deformations of sphero-conical coordinates or quasi-sphero-conical coordinates. We will show how,
starting from separation variables of the Neumann system, one gets an explanation of, in some sense, unusual choice of variables in \cite{BM2}.

It is well known that the Hamilton-Jacobi equation for a $n$-dimensional Neumann system
can be solved by the separation of variables in sphero-conical coordinates $u_1,...,u_{n-1}$ (see Moser \cite{Moser}).
Thus, using Theorem \ref{glavna}, after a time reparametrization, the rolling ball system separates  in coordinates $u_1,...,u_{n-1}$.
For the potential $V(x)=-(A^{-1} x,x)$
they are defined by the equations (see \cite{Moser}):
$$
\sum\limits_{i=1}^{n}\frac{x_i^2}{z-a_i^{-1}}=\prod\limits_{i=1}^{n-1}\frac{(z-u_i)}{U(z)}, \quad U(z)=\prod\limits_{j=1}^{n}(z-a^{-1}_j),
$$
and
$
a_1^{-1}<u_1<a_2^{-1}<u_2<\cdots a_{n-1}^{-1}<u_{n-1}<a_n^{-1}.
$
Formulas for $x_i$ are \cite{Moser}:
$$
x_i^2=\frac{(a_i^{-1}-u_1)(a_i^{-1}-u_2)\dots(a_i^{-1}-u_{n-1})}{(a_i^{-1}-a_1^{-1})(a_i^{-1}-a_2^{-1})\dots(a_i^{-1}-a_n^{-1})}, \quad i=1,\dots,n.
$$

Therefore, from \eqref{transf} one gets
\begin{equation}\label{sferokon}
\gamma_i^2=\frac{(a_i^{-1}-u_1)(a_i^{-1}-u_2)\dots(a_i^{-1}-u_{n-1})}{a_i\nu^2(a_i^{-1}-a_1^{-1})(a_i^{-1}-a_2^{-1})\dots(a_i^{-1}-a_n^{-1})}, \quad i=1,\dots,n,
\end{equation}
where
$
\nu^2=\langle A^{-1}x,x\rangle={\langle A\gamma,\gamma\rangle}^{-1}={a_1}^{-1}+{a_2}^{-1}+\dots+{a_n}^{-1}-u_1-u_2-...-u_{n-1}.
$

Let us consider briefly three-dimensional case and compare the above formulas with those from \cite{BM2}.
The operator $\mathbf{I}$ (see \eqref{spec-op}), after the identification $so(3)\cong \mathbb{R}^3$, corresponds to the matrix $\mathbf{J}=\diag(J_1,J_2,J_3)$ used in \cite{BM2} and we have
\begin{equation}\label{AuJ}
J_1=a_2 a_3, \quad J_2=a_1 a_3, \quad J_3=a_1 a_2.
\end{equation}

In \cite{BM2} separating coordinates $u,v$ are defined by the equation
\footnote{There is a typo in the equation \eqref{bmuv} in \cite{BM2}, where $\eta^2$ is missing. Further, to have a full correspondence with \cite{BM2}, $\gamma_i$ should be denoted by $n_i$ and $\eta^2$ by $\rho^2$.}
\begin{equation}
\frac{1}{\eta^2}\sum\limits_{i=1}^{3}\frac{\gamma_i^2}{(J_i-z)J_i}=\frac{(z-u)(z-v)}{J_1J_2J_3(J_1-z)(J_2-z)(J_3-z)},
\label{bmuv}
\end{equation}
where $0<J_1<u<J_2<v<J_3$,
and
$\eta^2=\langle \mathbf{J}^{-1}\gamma,\gamma\rangle$.
From \eqref{bmuv} it follows \cite{BM2}
\begin{equation}\label{borismam}
\begin{aligned}
\gamma_1^2&=\eta^2\frac{J_1(J_1-u)(J_1-v)}{(J_1-J_2)(J_1-J_3)},\\
\gamma_2^2&=\eta^2\frac{J_2(J_2-u)(J_2-v)}{(J_2-J_1)(J_2-J_3)},\\
\gamma_3^2&=\eta^2\frac{J_3(J_3-u)(J_3-v)}{(J_3-J_1)(J_3-J_2)},\quad
\eta^2 =(J_1+J_2+J_3-u-v)^{-1}.
\end{aligned}
\end{equation}

Using \eqref{AuJ} we get $\eta^2=(a_1a_2a_3)^{-1}\nu^{-2}$ and
the expressions \eqref{borismam} and \eqref{sferokon} (for $n=3$) coincide, where $u=a_1a_2a_3 u_1$ and $v=a_1a_2a_3 u_2$.
Geometrically, the classical sphero-conical coordinates $u_1,u_2$ are defined as the intersection of family of confocal conics
\begin{equation}
\mathcal Q_w:\quad \frac{x_1^2}{a_1^{-1}-w}+\frac{x_2^2}{a_2^{-1}-w}+\frac{x_3^2}{a_3^{-1}-w}=0
\label{sk}
\end{equation}
with the unit sphere $\langle x,x\rangle=1$.
The inverse of transformation \eqref{transf} together with $z=a_1a_2a_3w$ maps the intersection of the family of confocal conics \eqref{sk} with the sphere $\langle x,x\rangle=1$ to
the intersection of conics
\begin{equation}
\mathcal K_z: \quad \frac{a_1\gamma_1^2}{a_2a_3-z}+\frac{a_2\gamma_2^2}{a_1a_3-z}+\frac{a_3\gamma_3^2}{a_1a_2-z}=0.
\label{sk1}
\end{equation}
with the sphere $\langle \gamma,\gamma\rangle=1$. Therefore,
by the use of \eqref{AuJ}, one gets that coordinates $u,v$ in \eqref{bmuv} define conics $\mathcal K_u$, $\mathcal K_v$ from the family \eqref{sk1}, such that $\gamma\in \mathcal K_u\cap\mathcal K_v$.

Note that Ehlers and Koiller found the Chaplygin multiplier using the classical sphero-conical variables defined by the inertia operator of the ball, in which the system is not separable (for more details, see \cite{EK}).

\subsection{Noncommutative integrability of a symmetric case ($\epsilon=\pm 1$)}
For any pair of equal parameters $a_i=a_j$, the geodesic flow of the metric $g_*$ given by \eqref{th51} (for any value of the parameter $\epsilon$), has the additional Noether integral (e.g, see \cite{Jo7})
\begin{equation}\label{neter1}
\Phi_{ij}(\gamma',\gamma)=\gamma_i \frac{\partial L_*}{\partial \gamma_j'} - \gamma_j \frac{\partial L_*}{\partial \gamma_i'},
\end{equation}
i.e., the natural mechanical system \eqref{prirodni} preserves the function
\begin{equation}\label{neter2}
\tilde\Phi_{ij}\big(\frac{dx}{ds},x\big)=x_i\frac{dx_j}{ds}-x_j\frac{dx_i}{ds}.
\end{equation}

If we have at least three equal parameters, the systems are
integrable according to the non-commutative version of the
Liouville theorem. More precisely, assume $\{1,2,\dots,n\}=I_0\cup I_1 \cup \dots \cup I_r$, $I_i \cap I_j=\emptyset$,
\[
a_{i_0}=\alpha_0,  \, i_0 \in I_0,\quad a_{i_1}=\alpha_1, \, i_1\in I_1,\quad \dots \quad
a_{i_r}=\alpha_r, \,i_r \in I_r,
\]
$\alpha_i\ne\alpha_j, i\ne j$.
Then the geodesic flow of \eqref{th51} and the system \eqref{prirodni} are $SO(\vert I_0\vert)\times \dots \times SO(\vert I_r\vert)$--symmetric with Noether's integrals \eqref{neter1} and \eqref{neter2},
respectively. For $\epsilon=\pm 1$ they
are completely integrable in a non-commutative
sense by means of Noether's integrals and commuting integrals that are certain limits of integrals of a non-symmetric case
(e.g, see \cite{DDB, JB}, where a detail analysis for natural mechanical systems on a symmetric ellipsoid is given).
The corresponding Hamiltonian flows on the cotangent bundle of a sphere $S^{n-1}$ are generically quasi-periodic over invariant isotropic tori of dimension
\[
N=r+\sharp\{I_i\, \vert \vert I_i\vert \ge 2\}.
\]

\begin{remark}
For $a_i=a_j$, the function \eqref{neter1} in the original time $t$ takes the form
\[
\phi_{ij}(\dot\gamma,\gamma)=\nu(\gamma)\big(\gamma_i \frac{\partial L_{red}}{\partial \dot\gamma_j} - \gamma_j \frac{\partial L_{red}}{\partial \dot\gamma_i}\big)=
\frac{a_i}{\epsilon}(A\gamma,\gamma)^{\frac{1}{2\epsilon}}\big(\gamma_i\dot\gamma_j-\gamma_j\dot\gamma_i\big).
\]
One can relate the integral $\phi_{ij}$ of the reduced system \eqref{REDUCED2} with the $SO(n)$--reduction of the corresponding Noether function on the configuration space
$SO(n)\times S^{n-1}(g,\mathbf r)$ (e.g., see \cite{FNS, Jo8}).
\end{remark}

\subsection{Integrability for $\epsilon\ne \pm 1$}
Firstly note that in the case of $SO(n)$--symmetry, when $A$ is proportional to the identity matrix, the metrics \eqref{REDM0} and \eqref{th51} are proportional to the standard metric on a sphere and trajectories of \eqref{REDUCED} are great circles for all $\epsilon$.
Further, let us take $n=4$ and $a_1=a_2\ne a_3=a_4$. Then the complete
set of commuting integrals of the geodesic flow of \eqref{th51} is $\Phi_{12}$, $\Phi_{34}$, and the kinetic energy (the Lagrangian) $L_*=\frac12\langle\gamma',\gamma'\rangle_*$. 
Similarly, in the case of $SO(l)\times SO(n-l)$--symmetry, we have foliation on three-dimensional tori (or two-dimensional tori when $l=1$) and noncommuative integrability of the geodesic flow of the metric \eqref{th51} by means of Noether integrals \eqref{neter1} and the kinetic energy $L_*$ for any parameter $\epsilon$.

\begin{thm} \label{simetricna}
For the inertia operator \eqref{spec-op}, where
\begin{equation}\label{sing}
a_1=a_2=\dots=a_{l} \ne a_{l+1}=a_{l+2}=\dots=a_n,
\end{equation}
the reduced system \eqref{REDUCED} is integrable: generic motions, up to a time reparametrisation, are quasi periodic over  three dimensional invariant tori. For $l=1$ or $l=n-1$, the invariant tori are two-dimensional.
\end{thm}

In particular, the problem of rolling of a dynamically symmetric ball without spinning and twisting over a sphere in three dimension is integrable
(see \cite{BMB}). Generally, for $n\ge 4$, the operator \eqref{spec-op}
is not a physical inertia operator of a multidimensional rigid
body that has the form
\begin{equation}\label{manakov}
\mathbb{I}\omega=J\omega+\omega J,
\end{equation}
where $J$ is a symmetric positive definite matrix
called the \emph{mass tensor}
(e.g., see \cite{FeKo}). However, by taking $l=n-1$ in \eqref{sing}
and the conditions $a_1^2>D$, $2a_na_1>a_1^2+D$, we get the operator \eqref{manakov},
 where
\begin{equation*}
J=\diag(J_1,J_1,\dots,J_1,J_n), \quad J_1=\frac{a_1^2-D}{2}, \quad
J_n=a_1a_n-\frac{a_1^2+D}{2},
\end{equation*}
representing a $SO(n-1)$--symmetric rigid body (\emph{multidimensional Lagrange top} \cite{Be}).

Finally, we mention the case of the integrability of the Veselova problem with a physical inertia operator \eqref{manakov}, where
\[
J=\diag(J_1,\dots,J_1,J_2,\dots,J_2),
\]
 without involving Chaplygin Hamiltonisation
recently obtained in \cite{FaNaMo}. It would be interesting to consider the reduced equations \eqref{REDUCED} for the given inertia operator as well.

\subsection*{Acknowledgments}
We are grateful to Francesco Fasso, Luis C. Garcia-Naranjo, and James Montaldi for providing a preprint \cite{FaNaMo}
based on their talk given at \emph{Dynamics and integrability of nonholonomic and other non-Hamiltonian systems} (January 24-27, 2018, Padova),
to Alain Albouy for a stimulating discussion and pointing to us the references \cite{LA} and \cite{Th} during the conference GDIS 2018 (June 5--9, 2018, Moscow), as well as to the referees for helping us to improve the exposition of the results.
The research was supported by the Serbian Ministry of
Science Project 174020, \emph{Geometry and Topology of Manifolds,
Classical Mechanics and Integrable Dynamical Systems}.

\end{document}